\documentclass[11pt,a4paper]{article}
 
\usepackage[utf8]{inputenc}
\usepackage{authblk}

\usepackage{fullpage}

\usepackage{tabularx}

\usepackage{amsthm}
\usepackage{amsmath}
\usepackage{amssymb}
\usepackage{amsfonts}
\usepackage{mathtools}
\usepackage{enumitem}
\usepackage[pagebackref]{hyperref}
\usepackage[sort,numbers]{natbib}

\usepackage{makecell}

\usepackage{algorithm2e}
\DontPrintSemicolon

\usepackage{tikz}

\usepackage[capitalize]{cleveref}
\usepackage[switch]{lineno}

\usepackage[linecolor=green!70!black, backgroundcolor=green!10, bordercolor=black, textsize=scriptsize, obeyFinal
]{todonotes}

\newtheorem{theorem}{Theorem}
\newtheorem{lemma}[theorem]{Lemma}

\newtheorem{corollary}[theorem]{Corollary}

\theoremstyle{definition}

\newcommand{\problemdef}[3]{
	\begin{center}\fbox{
	\begin{minipage}{0.95\textwidth}
		\noindent
		#1
		\vspace{5pt}\\
		\setlength{\tabcolsep}{3pt}
		\begin{tabularx}{\textwidth}{@{}lX@{}}
			\textrm{Input:}     & #2 \\
			\textrm{Task:}  & #3
		\end{tabularx}
	\end{minipage}}
	\end{center}
}

\DeclarePairedDelimiterX{\abs}[1]{\lvert}{\rvert}{#1}

\newcommand{\bigO}{{O}}

\newcommand{\ES}{\textsc{Equitable Scheduling}\xspace}
\newcommand{\ESUP}{\textsc{Equitable Scheduling with Unit Processing Times}\xspace}
\newcommand{\ESSD}{\textsc{Equitable Scheduling with Single Deadlines}\xspace}
\newcommand{\ESSDabbrv}{\textsc{ESSD}\xspace}
\newcommand{\ESSDabbrvstar}{\textsc{ESSD*}\xspace}
\newcommand{\ESPC}{\textsc{Equitable Scheduling with Precedence Constraints}\xspace}
\newcommand{\ESPCabbrv}{\textsc{ESPC}\xspace}
\newcommand{\ESPCabbrvstar}{\textsc{ESPC*}\xspace}
\newcommand{\equitable}[1]{$#1$-equitable\xspace}

\newcommand{\gammaspecial}{\gamma^{\le n - \alpha}}
\newcommand{\xspecial}{x^{\le n - \alpha}}

\title{Equitable Scheduling on a Single Machine\footnote{An extended abstract of this paper appeared in the proceedings of the 35th AAAI Conference on Artificial Intelligence (AAAI '21)~\cite{BKN21}. This version contains full proof details and additional hardness results.}}

\author[1]{Klaus~Heeger\thanks{Supported by DFG Research Training Group 2434 ``Facets of Complexity''.}}
\author[2]{Danny~Hermelin}
\author[3]{George~B.~Mertzios\thanks{Supported by the EPSRC grant EP/P020372/1 and by DFG RTG~2434 while visiting TU~Berlin.}}
\author[2]{Hendrik~Molter\thanks{Supported by the DFG, project MATE (NI 369/17), and the ISF, grant No.~1070/20..}} 
\author[1]{Rolf~Niedermeier}
\author[2]{Dvir~Shabtay}

\affil[1]{\small Algorithmics and Computational Complexity, Faculty~IV, TU Berlin, Berlin, Germany\\
 \texttt{heeger@tu-berlin.de, rolf.niedermeier@tu-berlin.de}}
 \affil[2]{\small Department of Industrial Engineering and Management, Ben-Gurion~University~of~the~Negev, 
Beer-Sheva, 
Israel\\ \texttt{hermelin@bgu.ac.il, molter@post.bgu.ac.il, dvirs@bgu.ac.il}}
\affil[3]{\small Department of Computer Science, Durham University, UK\\ \texttt{george.mertzios@durham.ac.uk}}

\date{}

\begin{document}

\maketitle

\begin{abstract}
We introduce a natural but seemingly yet unstudied generalization of the problem of scheduling jobs on a single machine so as to minimize the number of tardy jobs. Our generalization lies in simultaneously considering several instances of the problem at once. In particular, we have $n$ clients over a period of $m$ days, where each client has a single job with its own processing time and deadline per day. Our goal is to provide a schedule for each of the $m$ days, so that each client is guaranteed to have their job meet its deadline in at least $k \leq m$ days. This corresponds to an
equitable schedule where each client is guaranteed a minimal level of service throughout the period of $m$ days. We provide a thorough analysis of the computational complexity of three main variants of
this problem, identifying both efficient algorithms and worst-case intractability results. 

\smallskip

\end{abstract}

\section{Introduction}

One of the most basic and fundamental scheduling problems is that of
minimizing the number of tardy jobs on a single machine. In this problem we are
given $n$~jobs, where each job~$j$ has an integer processing time~$p_j$ and an
integer deadline~$d_j$, and the goal is to find a permutation of the jobs so that the number of jobs exceeding their deadlines is minimized (a job $j$ exceeds its deadline if the total processing time of jobs preceding it in the schedule, including itself, is larger than~$d_j$). This problem is known as the $1||\sum U_j$ problem in the classical three-field notation for scheduling problems by~\citet{Graham79}. It is well-known that $1||\sum U_j$ is solvable in $O(n\log n)$~time~\cite{Maxwell1970,Moore,Strum1970}, but becomes e.g.\ NP-hard in case of simple (chain) precedence constraints even if all processing times~$p_j$ are the same~\cite{LENSTRA1980270}. There is also a more recent survey concerning the minimization of  the weighted number of tardy jobs~\cite{Adam} and the problem has also been thoroughly studied for parallel machines~\cite{Baptiste}.

Due to the ever increasing importance of high customer satisfaction, 
fairness-related issues are becoming more and more important in all
areas of resource allocation~\cite{BKN18,FSTW19,KK06,LR16,Wal20}.\footnote{For instance, in~2018, ACM
started its new conference series on ``Fairness, Accountability, and Transparency (originally FAT, since~2021 FAccT)''.} 
For instance, in their seminal work~\citet{BCPV96}
introduced the concept of proportionate progress, a fairness
concept for resource allocation problems. They applied
it to periodic scheduling by assigning resources to jobs according
to their rational weights between 0 and~1, thereby aiming to make sure that a job never gets an entire slot (in the periodic schedule) ahead or behind.
Nowadays, equity and fairness in resource allocation is a widely discussed topic, leading to considerations such as the ``price of
fairness''~\cite{BFT11} or to discussions about the abundance of
fairness metrics~\cite{GJRYZ20}.

We study a very natural but seemingly novel extension of the $1||\sum U_j$~problem, taking into account a very basic aspect of
equity among the customers in order to guarantee high customer satisfaction. Our task is to serve 
$n$ clients for $m$ days, where each client has a single job to be scheduled for every day. This should be done in an equitable fashion. We focus on a very simple notion of equity
where, given an integer parameter $k$, we request that each client receives satisfactory service in at least $k$ out of the $m$ days. In what follows, we refer to~$k$ as \emph{equity parameter}. It is important to note here that since all scheduling requests for all days and clients are assumed to be known in advance, we consequently still face an offline scenario of scheduling.

Consider the following motivating example. Imagine that a research group with $n$~PhD~students owns a single compute server, where each student has to submit a plan for their experiments for the next~$m$ days. Typically, at most one hour per experiment suffices and the needed computation time is known in advance. To make sufficient progress on their research, all students need to have regular access to the compute server for performing individual experiments. All students request access to the server every day, but due to high demand not all of them may be scheduled early enough during the day so that the experiments can still be evaluated the very same day. The chair of the group  wants to guarantee a schedule for the compute server such that every student can evaluate their experiments on the same day in at least $k$ out of~$m$ days. This is precisely the scenario we wish to model in the current manuscript. 

\subsection{Three Equitable Scheduling Variants}

Our model can be formally described as follows: We wish to schedule the jobs of a set of $n$ clients over $m$~days in an equitable way. At each day, each client has a single job to be scheduled non-preemptively on a single machine. 
We let $p_{i,j}$ and $d_{i,j}$ respectively denote the integer \emph{processing time} and \emph{deadline} of the job of client $i \in \{1,\ldots,n\}$ at day $j \in \{1,\ldots,m\}$. In addition, we let~$k$~denote an \emph{equity parameter} given as input, with $k \in \{0,\ldots,m\}$. 

A schedule $\sigma_j$ for day $j \in \{1,\ldots,m\}$ is a permutation $\sigma_j: \{1,\ldots,n\} \to \{1,\ldots,n\}$ representing the order of jobs to be processed on our single machine on day $j$. For a given schedule~$\sigma_j$, the \emph{completion time} $C_{i,j}$ of the job of client $i$ is defined as $C_{i,j}=\sum_{\sigma_j(i_0) \leq \sigma_j(i)} p_{i_0,j}$. In this way, the job meets its deadline on day $j$ if $C_{i,j} \leq d_{i,j}$. If this is indeed the case, then we say that client $i$ is \emph{satisfied} on day $j$, and otherwise $i$~is \emph{unsatisfied}. Our goal is to ensure that each client is satisfied in at least $k$ days out of the entire period of $m$ days; such a solution schedule (for all~$m$~days) is referred to as \emph{$k$-equitable}. Thus, depending on how large $k$ is in comparison with~$m$, we ensure that no client gets significantly worse service than any other client. 
 
\problemdef{\ES(ES):}
{A set of $n$ clients, each having a job with processing time $p_{i,j}$ and deadline $d_{i,j}$ for each day  $j \in \{1,\ldots,m\}$, and an integer $k$.}
{Find a set of $m$ schedules $\{\sigma_1,\ldots,\sigma_m\}$ so that for each $i \in \{1,\ldots,n\}$ we have $|\{ j \mid 1 \leq j \leq m \land C_{i,j} \leq d_{i,j}\}| \geq k$.}

We consider three variants of \ES, each corresponding to a well-studied variant of the $1||\sum U_j$ problem when restricted to a single day.
\begin{itemize}
\item In the first variant, which we call \ESUP (ESUP), the processing time of all jobs are unit in each day. That is, we have $p_{i,j} = 1$ for each $i \in \{1,\ldots,n\}$ and $j \in \{1,\ldots,m\}$. 

\item In the \ESSD (ESSD) problem, all jobs have the same deadline at each day. That is, at each day $j \in \{1,\ldots,m\}$ we have $d_{i,j} = d_j$ for each $i \in \{1,\ldots,n\}$. 

\item In the final variant, called \ESPC{} (ESPC), all processing
times are unit, and jobs share the same deadline at each day, \emph{i.e.}
$p_{i,j} = 1$ and $d_{i,j} = d_j$ for each $i \in \{1,\ldots,n\}$ and $j \in \{1,\ldots,m\}$. In addition, in each day we are given a \emph{precedence} directed acyclic graph (DAG) $G_j = (\{1,\ldots,n\},E_j)$ which represents precedence constraints on the jobs at day~$j$. We say that a schedule $\sigma_j$ is \emph{feasible} if for each $(i_1,i_2) \in E_j$ we have $\sigma_j(i_1) < \sigma_j(i_2)$. 
\end{itemize}

For each of these variants, we will also consider the special case where the input for each day is the same, and we will append an `*' to the name of the problem variant to indicate that this is the case we are considering. For example, in the ESPC* problem we have $d_{j_1}=d_{j_2}$ and $G_{j_1}=G_{j_2}$ for all $j_1,j_2 \in \{1,\ldots,m\}$.

We are mainly interested in exact
algorithms or algorithms with provable approximation guarantees.
We study the (parameterized) algorithmic complexity of all three
main variants (and some further variations) discussed above.

We use the following concepts from parameterized complexity theory~\cite{DF13,Nie06,FG06,Cyg+15}.
A parameterized problem~$L\subseteq \{(x,k)\in \Sigma^*\times \mathbb N\}$ is a subset of all instances~$(x,k)$ from~$\Sigma^*\times \mathbb N$,
where~$k$ denotes the \emph{parameter}.
A parameterized problem~$L$ is 
\begin{itemize}
\item FPT (\emph{fixed-parameter tractable}) if there is an algorithm that decides every instance~$(x,k)$ for~$L$ in~$f(k)\cdot |x|^{O(1)}$ time, and
 \item contained in the class XP if there is an algorithm that decides every instance~$(x,k)$ for~$L$ in~$|x|^{f(k)}$ time, 
\end{itemize}
where~$f$ is any computable function only depending on the parameter.
If a parameterized problem $L$ is W[1]-hard, then it is presumably not
fixed-parameter tractable.

\subsection{Our Results}
 Our main findings are as
follows.
\begin{itemize}
	\item For ESUP we show that the problem can be solved in polynomial time by a reduction to the \textsc{Bipartite Maximum Matching} problem. Our reduction can also be applied when jobs have release times, and when there is fixed number of machines available on each day.
	\item For \ESSDabbrv and \ESSDabbrvstar we show strong NP-hardness and W[1]-hardness for the parameter number $m$ of days even if $k=1$. On the positive side, we show that \ESSDabbrv can be solved in pseudo-polynomial time if the number $m$ of days is constant and is in FPT for the parameter number $n$ of clients. For \ESSDabbrvstar we give a polynomial-time algorithm that, for any~$k$, computes a $2k$-equitable set of schedules if there exists a $3k$-equitable set of schedules.
	\item For \ESPCabbrv we show NP-hardness even if~$k=1$ and there are only two days. For \ESPCabbrvstar we show NP-hardness and W[1]-hardness for the parameter number $m$ of days even if~$k=1$ and the precedence DAG only consists of disjoint paths. For \ESPCabbrv we also show NP-hardness for $k=1$ if the precedence DAG either consists of a constant number of disjoint paths or disjoint paths of constant length. On the positive side, we can show that \ESPCabbrv is in FPT for the parameter number $n$ of clients.
\end{itemize}

\section{Unit Processing Times}

In this section, we show that \ESUP can be solved in polynomial time by a reduction to the \textsc{Bipartite Maximum Matching} problem. Later in the section we will show that our reduction can also be applied when jobs have release times, and when there is fixed number of machines available on each day.

Recall that $p_{i,j}$ and $d_{i,j}$ respectively denote the processing time and deadline of the job of client $i$ on day~$j$, and that $k$ is the equity parameter. Let $d^*_j = \max_{1 \leq i \leq n} d_{i,j}$ denote the maximal deadline on day $j \in \{1,\ldots,m\}$. We create an undirected graph $G$ with the following vertices:
\begin{itemize}
\item For each $i\in \{1,\ldots,n\}$ and each $j\in \{1,\ldots,m\}$, we create a vertex $v_{i,j}$. The set of vertices $V = \{v_{i,j} : 1 \leq i \leq n, 1 \leq j \leq m\}$ represents all input jobs of all clients.
\item For each $d\in \{1,\ldots,d^*_j\}$ and each $j\in \{1,\ldots,m\}$, we create a vertex $u_{d,j}$. The set $U=\{u_{d,j}: 1 \leq d \leq d^*_j, 1\leq j \leq m\}$ represents all possible completion times of the all input jobs that meet their deadline.
\item For each $i\in \{1,\ldots,n\}$ and each $j \in \{1,\ldots,m-k\}$, we create a vertex $w_{i,j}$. The set $W = \{w_{i,j} : 1 \leq i \leq n, 1 \leq j \leq m-k\}$ represents the set of jobs that exceed their deadline.
\end{itemize}
The edges of $G$ are constructed as follows. For each $i\in \{1,\ldots,n\}$ and each $j\in \{1,\ldots,m\}$ we connect $v_{i,j}$ to:
\begin{itemize}
\item vertices $w_{i,1},\ldots, w_{i,m-k}$, and
\item vertices $u_{1,j}, \ldots, u_{d,j}$, where $d=d_{i,j}$.
\end{itemize}

\begin{lemma}\label{lem:matchingcorrect}
$G$ has a matching of size $nm$ if and only if there exists schedules $\{\sigma_1,\ldots,\sigma_m\}$ where no client is unsatisfied in more than $m-k$ days.  
\end{lemma}

\begin{proof}
$(\Leftarrow)$: Let $\{\sigma_1,\ldots,\sigma_m\}$ be a set of schedules where no client is unsatisfied on more than $m-k$ days.
Consider the job of client $i$ on day $j$, for some $i \in \{1,\ldots,n\}$ and
$j \in \{1,\ldots,m\}$. Note that the completion time of this job is
$C=C_{i,j}=\sigma_j(i)$. If $C \leq d_{i,j}$, then there is an edge
$\{v_{i,j},u_{C,j}\}$ in $G$, and we add this edge to the matching. If $C >
d_{i,j}$, then client $i$ is unsatisfied on day $j$. Let $\ell$ denote the
number of days prior to $j$ that client $i$ is unsatisfied. Then $\ell < m-k$,
since otherwise client $i$ would be unsatisfied in more than $m-k$ days
(including day $j$). We add the edge $\{v_{i,j},w_{i,\ell+1}\}$ to the matching.
In total, this gives us a matching of size $nm$ in $G$.

$(\Rightarrow)$: Assume that $G$ contains a matching of size $nm$. We create a set of schedules $\{\sigma_1,\ldots,\sigma_m\}$ accordingly. First note that the fact that $G$ is bipartite with one part being $V$, and $|V|=mn$, implies that every vertex in $V$ has to be matched. Let $v_{i,j}\in V$. Then this vertex is either matched to a vertex in $U$ or a vertex in $W$.
Because $N (u_{d + 1,j}) \subseteq N (u_{d,j})$ holds for every $j \in \{1,  \ldots, m\}$ and $d \in \{1, \ldots, d_j^* -1\} $, we may assume that for every $j \in \{1, \ldots, m\}$, there exists some $d_j$ such that vertex $u_{d, j}$ is matched if and only if $d \le d_j$.
\begin{itemize}
\item Suppose that $v_{i,j}$ is matched to some $u_{d,j_0}\in U$. Observe that $j=j_0$ and $d \leq d_{i,j}$ by construction of $G$. We set $\sigma_j(i)=d$, and so client $i$ is satisfied on day $j$. Observe that the fact that $u_{d,j_0}$ cannot be matched to any other vertex in $V$ guarantees that $\sigma_j(i) \neq \sigma_j(i_0)$ for any $i_0 \neq i$. Let $s_j$ denote the number of clients satisfied by $\sigma_j$ in this way.  
\item Suppose that $v_{i,j}$ is matched to some vertex $w_{i_0,j_0}\in W$. Note
that $i=i_0$ and $j_0 \leq m-k$ by construction of $G$. Let $x_{i,j} = |\{i_0 < i : w_{i_0,j} \text{ is matched }\}|$. Then we set $\sigma_j(i)=s_j+x_{i,j}+1$.
\end{itemize}
After each $\sigma_j$ is permutation from $\{1,\ldots,n\}$ to $\{1,\ldots,n\}$, and no client is unsatisfied in more than $m-k$ days under $\{\sigma_1,\ldots,\sigma_m\}$.
\end{proof}

Since we may assume that $d^*_j \le n$ for every $j\in \{1, \ldots, m\}$, we observe that $G$ has $O(mn)$ vertices and $O(mn^2 + m^2n)$ edges, and it can be constructed in $O(mn^2 + m^2n)$ time. Using the algorithm of \citet{hopcroft1973n} for \textsc{Bipartite Maximum Matching}, this gives us the following:

\begin{theorem}
\label{thm:unitprocgeneral}%
ESUP can be solved in $O((n+m)\cdot(nm)^\frac{3}{2})$ time.
\end{theorem}

We remark that this algorithm is very flexible and can easily be extended to the
setting where the jobs have release dates and where there are multiple parallel
machines.
The main idea is that the vertices in $U$ represent time slots for
jobs and a job has an edge to all time slots that are before the job's deadline.
To incorporate release dates we additionally remove edges to time slots that are
``too early''. Finally, to model parallel machines, we essentially introduce
copies of the vertices in $U$ for each of the machines.
 We give a formal description on how to construct the graph $G$ in which we
compute a maximum bipartite matching in the following.

Recall that $p_{i,j}$ and $d_{i,j}$ respectively denote the processing time and
deadline of the job of client $i$ on day~$j$, and that $k$ is the equity
parameter. 
Let $r_{i,j}< d_{i,j}$ the release date of the job of client~$i$ on
day $j$ and $x_j$ the number of parallel machines available on day $j$.
Let $d^*_j = \max_{1 \leq i \leq n} d_{i,j}$ denote the maximal deadline on day
$j \in \{1,\ldots,m\}$ and let $x^* = \max_{1 \leq j \leq m} x_j$ denote the
maximal number of machines available on a day.
We create an undirected graph $G$ with the following vertices:
\begin{itemize}
\item For each $i\in \{1,\ldots,n\}$ and each $j\in \{1,\ldots,m\}$, we create
a vertex $v_{i,j}$. The set of vertices $V = \{v_{i,j} : 1 \leq i \leq n, 1
\leq j \leq m\}$ represents all input jobs of all clients.
\item For each $d\in \{1,\ldots,d^*_j\}$, each $j\in \{1,\ldots,m\}$, and each
$x\in\{1,\ldots,x^\star\}$, we create a vertex $u_{d,j,x}$. The set
$U=\{u_{d,j,x}\mid  1 \leq d \leq d^*_j, 1\leq j \leq m, 1\le x\le x^\star\}$
represents all possible completion times on some machine of the all input jobs
that meet their deadline.
\item For each $i\in \{1,\ldots,n\}$ and each $j \in \{1,\ldots,m-k\}$, we
create a vertex $w_{i,j}$. The set $W = \{w_{i,j} : 1 \leq i \leq n, 1 \leq j
\leq m-k\}$ represents the set of jobs that exceed their deadline.
\end{itemize}
The edges of $G$ are constructed as follows. For each $i\in \{1,\ldots,n\}$ and
each $j\in \{1,\ldots,m\}$ we connect $v_{i,j}$ to:
\begin{itemize}
\item vertices $w_{i,1},\ldots, w_{i,m-k}$, and
\item vertices $u_{r,j,1}, \ldots, u_{d,j,x}$, where $d=d_{i,j}$, $r=r_{i,j}+1$,
and $x=x_j$.
\end{itemize}

The proof of correctness is analogous to the proof of
Lemma~\ref{lem:matchingcorrect}. We omit the details.

\section{Single Deadline on Each Day}
In this section, we investigate the computational complexity of \ESSD.

\subsection{Hardness Results}
We first show that \ESSDabbrv is NP-hard even if all numbers involved are small constants.
\begin{theorem}\label{thm:ESSDhardness}
  \ESSDabbrv is NP-hard even if ${k=1}$ and ${d=3}$.
\end{theorem}

\begin{proof}
  We present a reduction from \textsc{Independent Set} on 3-regular graphs. In this problem we are given a graph $G=(V,E)$ where every vertex has degree 3 and an interger $\ell$ and are asked whether $G$ contains an independent set of size $\ell$. This problem is known to be NP-hard~\cite{GJS74}.
  Let $(G=(V,E), \ell)$ be an instance of \textsc{Independent Set}, where $G$ is an $3$-regular graph.

  We construct an instance of \ESSDabbrv as follows. We set $k=1$ and $d=3$.
  For each vertex~$v\in V$, we add a client~$c_v$, and for each edge $e\in E$, we add a client~$c_e$.
  There exist $m=2|V| -\ell$ many days.
  We order the vertices arbitrarily, i.e., $V = \{v_1, \dots, v_n\}$.
  On day $i$, the job of client~$c_{v_i}$ has processing time $3$, while the jobs of client $c_{v}$ with $v\neq v_i$ have processing time~$4$.
  The jobs of clients $c_e$ have processing time $1$ if $e$ is incident to $v_i$ and processing time $4$ otherwise.
  Finally, on days $|V| + 1$ to $2|V| - \ell$, jobs of clients $c_v$ for $v\in V$ have processing time $3$, and jobs of clients $c_e$ for $e\in E$ have processing time $4$.

  $(\Rightarrow)$: An independent set $I$ of size $\ell$ implies a feasible schedule by scheduling the jobs of the $|V| -\ell$ clients~$c_v$ with $v\in V\setminus I$ to the days $|V | + 1$ to $2|V| - \ell$.
  For each $v_i\in I$, we schedule the job of client $c_{v_i}$ to Day $i$.
  For all clients $c_e$ we schedule one of their jobs on a day corresponding to an arbitrary endpoint of $e$ that is not contained in the independent set~$I$.

  $(\Leftarrow)$:  Note that every feasible schedule must schedule a job of client $c_{v_i}$ to day $i$ for at least~$\ell$ different clients $c_{v_i}$.
  We claim that $I := \{v_i \in V \mid c_{v_i} \text{ has a job}$ $\text{that is scheduled on day } i\}$ is an independent set.
  If there exists an edge $e= \{v_i, v_{i'}\}$ for $v_i, v_{i'}\in I$, then this edge can be scheduled neither on a day $i$ nor on day $i'$ and therefore cannot be scheduled on any day.
  Thus, these at least $\ell$ vertices in $I$ form an independent set.
\end{proof}

We can further show that \ESSDabbrvstar (i.e., \ESSDabbrv where the processing time of the job of each client is the same every day) is NP-hard and W[1]-hard when parameterized by the number of days.\footnote{Parameterized complexity studies of NP-hard scheduling problems gained more and more interest over the recent years~\cite{BBGKN21,BBN19,GHM20,HMPSY20,HPST19,HST19,MB18}; we contribute to this field with several of our results.}
\begin{theorem}\label{thm:sameduedatehardness}
\ESSDabbrvstar is NP-hard and W[1]-hard when parameterized by the number $m$ of days even if $k=1$ and all numbers are encoded in unary.
\end{theorem}
\begin{proof}
We present a parameterized reduction from \textsc{Unary Bin Packing}, where given a set $I=\{1, \ldots, n\}$ of items with sizes $s_i$ for $i\in I$, $b$ bins of size $B$, we are asked to decide whether it is possible to distribute the items to the bins such that no bin is overfull, i.e., the sum of the sizes of items put into the same bin does not exceed $B$. \textsc{Unary Bin Packing} is known to be NP-hard and W[1]-hard when parameterized by $b$~\cite{jansen2013bin}. Given an instance of \textsc{Unary Bin Packing}, we construct an instance of \ESSDabbrvstar as follows.

We set the number of days to $b$, i.e., $m=b$. For each item $i\in I$ and each day $j\in\{1, \ldots, m\}$ we create a job for client~$i$ with processing time $p_{i,j}=s_i$, i.e., the processing time is the same every day. The deadline for every day is $d=B$. Finally, we set $k=1$. This finishes the construction which can clearly be done in polynomial time.

$(\Rightarrow)$: Assume the \textsc{Unary Bin Packing} instance is a YES-instance. Then there is a distribution of items to bins such that no bin is overfull. If item $i\in I$ is put into the $j$th bin for some $j\in \{1,\ldots, b\}$, then we schedule the job of of client $i$ on day $j$ to be processed. Since every item is put into one bin, every client has a job that is scheduled to be processed at one day and since no bin is overfull, all scheduled jobs can be processed before their deadline $d$. It follows that we have a \equitable{1} set of schedules.

$(\Leftarrow)$: Assume we have a \equitable{1} set of schedules. Then every client has a job on at least one day that is processed. Let client $i$ have a job that is processed on day $j$. Then we put item~$i$ into the $j$th bin. Since the processing time $p_{i,j}$ is the same as the size of item $i$ and the sum of the processing times of jobs that are scheduled to be processed on the same day is at most $d=B$, the sum of sizes of items that are put into the same bin is at most $B$. Hence, we have a valid distribution of items into bins.
\end{proof}

\subsection{Algorithmic Results}
We first show that we can solve \ESSDabbrv in pseudo-polynomial time if the number of days $m$ is constant. Note that this implies that \ESSDabbrv is in XP when parameterized by the number of days if all processing times and the deadline is encoded unary. Theorem~\ref{thm:sameduedatehardness} shows that we presumably cannot expect to be able to obtain an FPT algorithm for this case.
\begin{theorem}\label{thm:DP1}
\ESSDabbrv can be solved in $O(d_{\max}^m\cdot \binom{m}{k}\cdot n)$ time, where $d_{\max}=\max_j d_j$.
\end{theorem}
\begin{proof}
We give a dynamic programming algorithm for this problem.
The table $T$ maps from $\{0,\ldots,n\}\times \{1,\ldots, d_{\max}\}^m$ to $\{\texttt{true}, \texttt{false}\}$ and intuitively an entry $T[i,b_1, \ldots, b_m]$ is true if and only if it is possible to schedule $k$ jobs of each client $\{1,\ldots,i\}$ such that the total processing time on day $j$ is at most $b_j$ for all days $j\in \{1,\ldots, m\}$.
Formally, the table is defined as follows.
\begin{align*}
T[0,b_1, \ldots, b_m] &= \texttt{true}\\
T[i,b_1, \ldots, b_m] &= 
\end{align*}
\[\bigvee_{\{x_1, \ldots, x_k\}\in\binom{\{1,\ldots, m\}}{k}} T[i-1,\ldots, b_{x_1}-p_{i,x_1},\ldots, b_{x_k}-p_{i,x_k},\ldots]\]
Intuitively, we ``guess'' on which days we want to schedule jobs of client $i$ and then look up whether there exists a set of \equitable{k} schedules for clients $\{1,\ldots,i-1\}$ to which we can add $k$ jobs of client~$i$ such that the total processing time comply with the upper bounds. 

It is easy to check that the input instance is a YES-instance if and only if $T[n,d_1,\ldots,d_m]=\texttt{true}$. The size of the table $T$ is in $O(d_{\max}^m\cdot n)$ and computing one table entry takes $O(\binom{m}{k})$ time. Hence, we arrive at the claimed running time.
\end{proof}

Next, we show that \ESSDabbrv can be solved in polynomial time if the number of clients $n$ is constant. In other words, we show that \ESSDabbrv is in XP when parameterized by the number of clients.
\begin{theorem}\label{thm:ESSDXP2}
\ESSDabbrv can be solved in $O((2k+2)^{n}\cdot m)$ time. 
\end{theorem}

\begin{proof}
We give dynamic programming algorithm for this problem. The table $T$ maps from $\{1,\ldots,m\}\times
\{0,\ldots,k\}^{n}$ to $\{\texttt{true}, \texttt{false}\}$. Entry $T[j,\ell_{1},\ldots,\ell_{n}]$ is true 
if and only if it is possible to provide schedules for days $\{1,\ldots ,j\}$ such that for each client $i$ we have that its job is processed on exactly $\ell_{i}$ days ($\ell_{i}\leq k$). We have that 
\begin{align*}
T[j,0,\ldots,0]&=\texttt{true} \text{ for all } j\in\{1,\ldots,n\}\text{, and}\\
T[j,\ell_{1},\ldots,\ell_{n}]&=\texttt{true}
\end{align*}
if there exists a subset of clients $I^*\subseteq \{1,\ldots,n\}$ such that 
\[
\sum_{i\in I^*}p_{i,j}\leq d_j
\]%
and that 
\[
T[j-1,\ell_{1}-\mathbb{I}_{1\in I^*},\ldots,\ell_{n}-\mathbb{I}_{n\in I^*}]=\texttt{true}
\]%
where $\mathbb{I}_{i\in I^*}$ is an indicator variable for the event that $i\in I^*$, i.e.,
\[
\mathbb{I}_{i\in I^*}=\left\{ 
\begin{array}{cc}
1\text{ } & \text{if } i\in I^* \\ 
0\text{ } & \text{otherwise}%
\end{array}%
\right. 
\]%

It is easy to check that the input instance is a YES-instance if and only if $T[m,k,\ldots,k]=\texttt{true}$. The theorem follows from the fact that we have $O(m\cdot (k+1)^{n})$ many $T[j,\ell_{1},\ldots,\ell_{n}]$ values to compute and $O(2^{n})$ possible $I^*$
subsets to check for calculating any $T[j,\ell_{1},\ldots,\ell_{n}]$ value.
\end{proof}

We now strengthen the above result by showing that \ESSDabbrv is in FPT when
parameterized by $n$. To do this, we give an integer linear programm formulation
for the problem and use a famous result by~\citet{lenstra1983integer}. Note,
however, that Theorem~\ref{thm:ESSDXP2} is a purely combinatorial result and that the
implicit running time of Theorem~\ref{thm:ESSDILP} is at least double exponential.

\begin{theorem}\label{thm:ESSDILP}
  \ESSDabbrv is in FPT when parameterized by the number of clients $n$.
\end{theorem}

\begin{proof}
  First we partition the days into equivalence classes.
  We say that two days $j$ and $j'$ are equivalent if for any subset~$S$ of clients all jobs of $S$ can be scheduled together on Day $j$ if and only if they can be scheduled together on Day $j'$.
  Let $\mathcal{E}$ be the set of equivalence classes.
  Clearly, $|\mathcal{E}| \le 2^{2^n}$.
  We write that $S\succ E$ for a set of clients $S$ and an equivalence class $E$ if the sum of the processing times of all jobs from $S$ exceeds the deadline on every day from $E$.

  We design an ILP with one variable $x_{E, S}$ for each pair of equivalence class $E\in \mathcal{E}$ and subset of clients $S$ from $E$ as follows.
  \begin{align*}
    x_{E, S} = 0 & \qquad \text{if } S\succ E\\
    \sum_{S: i \in S} \sum_{E\in \mathcal{E}} x_{E, S} \ge k & \qquad \forall i \in \{1, \dots, n\}\\
    \sum_{S \subseteq \{1, \dots, n\}} x_{E, S} = |E| & \qquad \forall E \in \mathcal{E}
  \end{align*}

  Since the number of variables is at most $2^n \cdot 2^{2^n}$, it follows by~\citet{lenstra1983integer} that the ILP can be solved in FPT-time parameterized by $n$.

  Given a solution to the ILP, we get a $k$-equitable schedule by scheduling for each variable~$x_{E,S}$ the jobs of $S$ on exactly $x_{E, S}$ days of the equivalence class $E$.
  By the third condition, this results in one schedule for every day.
  By the first condition none of the scheduled jobs is tardy.
  By the second condition, the schedule is $k$-equitable.

  Vice versa, given a $k$-equitable schedule, we construct a feasible solution to the ILP by setting $x_{E, S}$ to be the number of days from equivalence class $E$ scheduling exactly the jobs from~$S$ before the deadline.
  The first condition is then fulfilled by the definition of $S \succ E$.
  The second condition is fulfilled as the schedule is $k$-equitable.
  The third condition is fulfilled as there is exactly one schedule for each day.
\end{proof}

In the remainder of this subsection, we investigate the canonical optimization version of \ESSDabbrvstar where we want to maximize~$k$.
Note that the existence of a polynomial-time approximation algorithm with any factor (i.e., an algorithm computing a solution for an instance $\mathcal{I}$ of value $\text{ALG}(\mathcal{I})$ such that $f(\mathcal{I})\cdot\text{ALG}(\mathcal{I}) \ge \text{OPT}(\mathcal{I})$ for some function $f$) implies P $=$ NP, since distinguishing between the cases $k=0$ and $k=1$ is NP-hard (see Theorem~\ref{thm:sameduedatehardness}).

However, we will show that for any instance with optimal solution value $3k$, we can find a solution of value $2k$.
We make a case distinction on $k$: we first show an algorithm for that case that $k\le\frac{m}{2}$ and afterwards an algorithm for the case of $k>\frac{m}{2}$.

\begin{lemma}\label{lem:apx-small-k}
  Given a YES-instance $\mathcal{I} = (\{p_1, \ldots, p_n\}, m,d, k)$ with $k \le \frac{m}{2}$ of \ESSDabbrvstar, one can compute a solution to the instance $\mathcal{I} ' := (\{p_1, \ldots, p_n\}, m,d, k')$ with $k':=2\lfloor \frac{k}{3}\rfloor$ in $O (n \cdot (k + \log n)) $ time.
\end{lemma}

\begin{proof}
 We apply an algorithm similar to the so-called ``First-Fit-Decreasing'' algorithm for \textsc{Bin Packing}~\cite{JohnsonDUGG74}.
  Set~$k':=2\lfloor \frac{k}{3}\rfloor$. The algorithm works in the following steps.
\begin{enumerate}
\item Order the clients decreasingly by the processing time of their jobs.
\item Iterate through the clients in the computed order.

For each client we schedule $k'$ of their jobs on the first~$k'$ days that have enough space, i.e., after the jobs are scheduled the sum of processing times of the scheduled jobs for each day is at most $d$.

\textsl{Note that so far (i.e., without Step~3), the jobs of each client are scheduled in a block of~$k'$ consecutive days that starts a some day $j$ with $j\bmod k'=0$.}
\item If there is a client $i$ who cannot have $k'$ of its jobs scheduled that way, do the following:

\textsl{Note that when this happens for the first time, it means that all blocks of $k'$ consecutive days that starts a some day $j$ with $j\bmod k'=0$ are ``full''. We now make a case distinction on the number $m\bmod k'$ of days that are not part of any of these blocks.}
\begin{itemize}
\item If Step 3 is invoked for the first time, then let $i'$ be the client with smallest processing time scheduled on day $\lfloor \frac{2m}{3}\rfloor + 1$.
Let $j$ be the first day that has a job of client~$i'$ scheduled. Schedule jobs of clients $i$ and $i'$ to days $\{m-(m\bmod k')+1, \ldots, m-(m\bmod k')+\frac{k'}{2}\}$ and replace the jobs of client $i'$ that are scheduled on days $\{j, \ldots, j+\frac{k'}{2} - 1\}$ with jobs of client $i$.

If Step 3 is invoked for the second time, then output FAIL.
\item If $m\bmod k'< \frac{k'}{2}$, then output FAIL.
\end{itemize}
\item If all clients are processed, output the schedules.
\end{enumerate}
  
  We first show that if the presented algorithm outputs a set of schedules, the set is \equitable{k'}. If $m\bmod k'< \frac{k'}{2}$, then this is obvious. If $m\bmod k'\ge \frac{k'}{2}$, then we have to check that Step~3 of the algorithm does not produce infeasible schedules. Observe that in Step~3, we have that $p_{i'}\ge p_i$ since the clients are ordered by the processing time of their jobs and client $i'$ is processed before client $i$. This means that replacing a jobs of client $i'$ by a job of client $i$ on some day cannot violate the deadline unless it was already violated before swapping the jobs.
  Observe that if $\mathcal{I}$ is a YES-instance, then there can be at most $\lfloor \frac{m}{k} \rfloor$ jobs with processing time more than~$\frac{d}{2}$.
  Thus there are at most $\lfloor\frac{2m}{3} \rfloor$ days on which our algorithm schedules a job with processing time more than $\frac{d}{2}$.
  Since the algorithm processes the jobs in decreasing order, all jobs with length more than $\frac{d}{2}$ are scheduled only on the first $\lfloor \frac{2m}{3} \rfloor$ days.
  It follows that $p_{i'} \le \frac{d}{2}$ since it is scheduled on day~$\lfloor \frac{2m}{3} \rfloor + 1$.
  It follows that $p_i \le p_{i'} \le \frac{d}{2}$, and thus, the deadline is not violated on days $m- (m \mod k')+1, \dots, m - (m \mod k') + \frac{k'}{2}$.
  This implies that Step~3 always produces \equitable{k'} sets of schedules.

In the remainder of the proof we show that if the presented algorithm outputs FAIL, then~$\mathcal{I}$ is a NO-instance.
On an intuitive level, the main idea is to show that the first $\lceil\frac{2m}{3} \rceil$ days are ``full'' and the remaining $\lfloor\frac{m}{3} \rfloor$ days have at least $\lceil\frac{2m}{3} \rceil-k'$ jobs (in total) scheduled. This then allows us to show that the total processing time if $k$ jobs of each client were scheduled exceeds $m\cdot d$, which implies that~$\mathcal{I}$ is a NO-instance.
  
  Since all jobs with length more than $\frac{d}{2}$ are scheduled only on the first $\lfloor \frac{2m}{3} \rfloor$ days, it follows that if the algorithm outputs FAIL, then the last $\lceil\frac{m}{3} \rceil$ days have either at least two jobs scheduled or none.

  Assume that our algorithm outputs FAIL and let client $i^*$ be the client that was processed when the algorithm output FAIL.
  Note that there are strictly less than $\frac{k'}{2}$ days with no jobs scheduled, independent on whether $m \mod k' \le \frac{k'}{2}$.
Thus, among the last $\lfloor\frac{m}{3} \rfloor$ days, (strictly) less than $\frac{k'}{2}$ days have no jobs scheduled and all others have at least two jobs scheduled. Together with $k'$ jobs of client $i^*$ which are not scheduled at all, we have at least $2(\lfloor\frac{m}{3}\rfloor-\frac{k'}{2}+1)+k'\ge  2\lceil\frac{m}{3} \rceil$ jobs, all of which have a processing time of at least~$p_{i^*}$. Let the set of these jobs be called $J^*$. Since the jobs of client $i^*$ could not be scheduled in the first $\lceil\frac{2m}{3} \rceil$ days, we know that the total processing time of all jobs from one of the first $\lceil\frac{2m}{3} \rceil$ days plus $p_{i^*}$ or the processing time of any job in $J^*$ is larger than the deadline~$d$. Intuitively, this allows us to ``distribute'' the processing times of the jobs in $J^*$ to the first $\lceil\frac{2m}{3} \rceil$ days (note that $|J^*|\ge\lceil\frac{2m}{3} \rceil$) and derive the following estimate:
$k' \sum_{i : p_i\ge p_{i^*}} p_i > \lceil\frac{2m}{3} \rceil \cdot d$. Substituting $k'$ with $k$ and summing over all clients, we get $k \sum_{i \in \{1,\ldots,n\}} p_i > m \cdot d$, which is a contradiction to the assumption that $\mathcal{I}$ is a YES-instance.

Since First-Fit-Decreasing can be implemented in $O(n^*\log n^*)$~\cite{Johnson74}, where $n^*$ is the number of elements, Steps~1 and 2 can be performed in time $\bigO (n \log n + n \cdot k)$ by running First-Fit-Decreasing on the instance with one element of size $p_i$ for each client~$i$ and $\lfloor \frac{m}{k'}\rfloor$ bins of size~$d$, and then cloning the solution $k'$ times.
  Step~3 clearly runs in $\bigO(k)$ while Step~4 runs in constant time.
  The running time of $\bigO (n (k + \log n))$ follows.
\end{proof}

We now turn to the case $k > \frac{m}{2}$.

\begin{lemma}\label{lem:apx-large-k}
Given a YES-instance $\mathcal{I} = (\{p_1, \ldots, p_n\}, m,d, k)$ with $k > \frac{m}{2}$ of \ESSDabbrvstar, one can compute a solution to the instance $\mathcal{I} ' := (\{p_1, \ldots, p_n\}, m,d, k')$ with $k':=\lfloor \frac{2k}{3}\rfloor$ in $O (n \cdot (k + \log n)) $ time.
\end{lemma}

\begin{proof}
We classify the clients into two groups based on the processing time of their jobs:
A client $i$ is \emph{large} if $p_i > \frac{d}{3}$ and \emph{small} otherwise (i.e., if $p_i \le \frac{d}{3}$).
Set $k':=\lfloor \frac{2k}{3}\rfloor$.
We start with some basic obervations:
  \begin{enumerate}
    \item \label{item:two-large-jobs}
    There are no two clients $i_1$ and $i_2$ with $p_{i_1} + p_{i_2} > d$.
    Since $k > \frac{m}{2}$, every solution to $\mathcal{I}$ must schedule jobs of clients~$i_1$ and $i_2$ at least once to the same day by the pidgeon-hole principle, which is impossible if $p_{i_1} + p_{i_2} > d$.
    \item There are at most three large clients. Assume for contradiction that there are four large clients. Then, since $k > \frac{m}{2}$, by the pidgeon-hole principle there is one day that has three jobs from three of the four large clients scheduled, which is impossible since the total processing time on that day would exceed $d$.
    
    \item \label{item:overfull}
    The total processing time of all jobs that need to be scheduled cannot exceed $m\cdot d$, i.e., $k\sum_{i\in \{1,\ldots,n\}} p_i \le m\cdot d$.
    
    Note that this implies that if a \equitable{k'} set of schedules schedules on each day jobs with total processing time larger than $\frac{2d}{3}$, then $\mathcal{I}$ is a NO-instance, since then $k \sum_{i\in \{1,\ldots,n\}} p_i \ge \frac{3}{2}k' \sum_{i\in \{1,\ldots,n\}} p_i > \frac{3}{2} m \frac{2d}{3} = m\cdot d$.

  \end{enumerate}
From now on we assume that the first two observations hold, otherwise $\mathcal{I}$ is a NO-instance. 

Intuitively, we will mostly try to use the third observation to show that our
algorithm is correct: We greedily fill up all days with jobs until no job of a
small client fits in any day. If this happens and we do not have a
\equitable{k'} set of schedules, then by the third observation we can deduce
that we were facing a NO-instance. However, in order to do this, we first have
to deal with some special cases explicitely (which are handled in Steps 2 and
3 of the algorithm in the next paragraph).
If the total processing time of the jobs of all small clients is very small (i.e., at most $\frac{d}{3}$) we can
construct a \equitable{k'} set of schedules directly. We also need to treat some
cases where the total processing time of the jobs of all small clients is at
most $\frac{2d}{3}$ separately, hence then we can have the case that we cannot
schedule a job of any small client on a certain day and still the total
processing time on that day does not exceed $\frac{2d}{3}$, which prevents us
from applying the third obervation.
Formally, we sort once in all clients by the processing times of their jobs, and then we compute a set of schedules in the following way. 
\begin{enumerate}
\item If the sum of
processing times of all small clients is at most $\frac{d}{3}$ and there are at
least two large clients, then we do the following.

We schedule the jobs of
 the up to three large clients one after another in the following way.
 We pick the $k$ days having the most free processing time and schedule a job of client whose job we currently schedule on these days.
 If these schedules exceed the deadline on one day, then we output FAIL.
  
  Now we pick $\lceil\frac{k}{3} \rceil$ days where the first large client has
  a job scheduled, we remove that job and replace it with jobs of all small clients.
  Next, we pick $\lceil \frac{k}{3} \rceil$ \emph{different} days where the second large
  client has a job scheduled, we remove that job and replace it with jobs of all small clients.
\item If the sum of
processing times of all small clients is at most $\frac{2d}{3}$ and there are at most two large jobs, then we do the
following.
\begin{itemize}
  \item If there are no large clients, we schedule all jobs of all small clients
  on the first $k'$ days.
  \item If there is only one large client, then we
schedule the job of the large client on the first $k'$ days and on the $m-k'$
remaining days we schedule jobs of all small clients.

If $m < 2k'$, then we recursively find a \equitable{\lfloor \frac{2}{3} (k + k' - m)\rfloor} schedule for the small clients
on the first $k$ days where the deadline is set to $d-p_\ell$, where $p_\ell$ is the
processing time of the job of the large client.
\item If there are two large clients and $k'<\frac{m}{2}$, then we schedule
  jobs of the two large clients on the first $k'$ days and jobs of all small
  clients on the last $k'$ days.
\end{itemize}
  
  \item We schedule the jobs of
 the up to three large clients one after another in the following way.
 We pick the $k'$ days having the most free processing time and schedule a job of client whose job we currently schedule on these days.
 If these schedules exceed the deadline on one day, then we output FAIL.
  \item We schedule the jobs of the small clients of after another in the following fashion. We
  fix an order of the small clients and create a list $L$ repeating this order $k'$ times.
  We process the days from the first one to the last one as follows.
  Until the list $L$ gets empty, we schedule the job of the first client~$i$ in $L$ and delete (this appearance of) $i$ from $L$,
  unless the job of $i$ is already scheduled on this day, or the processing time of this job together with the processing time of all jobs already scheduled on this day exceeds the deadline.
  If the list $L$ is non-empty after we processed the last day, we return FAIL.
\end{enumerate}

Assuming the algorithm does not recurse in Step~2, it is easy to check that if the algorithm does not output FAIL, then we found
a \equitable{k'} set of schedules.

If the algorithm recurses in Step~2, then the large job is scheduled $k'$ times and every small job is scheduled $m - k' + \lfloor \frac{2}{3} (k + k' - m)\rfloor$ times.
Note that using $k' \le \min \{m, \frac{2k}{3}\}$, we get $m - k' + \frac{2}{3} ( k + k' -m) = \frac{m}{3} + \frac{2k}{3} - \frac{k'}{3} \ge \frac{k'}{3}+ k' - \frac{k'}{3} = k'$.
Thus, due to the integrality of $m$ and~$k'$, every small job is scheduled at least $k'$ times.

In the remainder of the proof we show that if the algorithm outputs
FAIL, then $\mathcal{I}$ is a NO-instance.

Assume that the algorithm outputs FAIL in Step~1. Since we assume that the
first basic observation holds this can only happen if there are three large
clients and the algorithm schedules one job of each large client to one day and
all other days have two jobs of large clients scheduled. 
This can only happen if $3k>2m$, however then, by the pidgeon-hole principle,
any feasible solution would have to schedule jobs of each of the three large
clients on the same day. This is a contradiction to the assumption that 
$\mathcal{I}$ is a YES-instance.

By the same argument, we have that if the algorithm outputs FAIL in Step~3, then
$\mathcal{I}$ is a NO-instance.

Now assume that the algorithm outputs FAIL in Step~4. 
Since Step~4 was applied, the sum of processing of all small clients is larger than $\frac{2d}{3}$, or the processing time of all small clients exceeds $\frac{d}{3}$ and on each day, at least one large job is scheduled.
This implies that for
each day the sum of processing times of jobs scheduled at that day is larger
than $\frac{2d}{3}$, since otherwise the algorithm would have scheduled the job from the next client in $L$.
However then, by the third basic observation, we know
that $\mathcal{I}$ is a NO-instance.

Finally, in Step~3, the algorithm may output FAIL only in the recursive call.
Fix a solution to $\mathcal{I}$ which schedules the large jobs on the first $k$ days.
Thus, on the first $k'$ days, this solution has to schedule each of the small jobs at least $m- k'$ times, and has only $d- p_\ell$ time on each of these days, where $p_\ell$ is the processing time of the large job.
It follows that the created instance admits a \equitable{k - (m-k')}-schedule.
Thus, by induction, our algorithm finds a \equitable{k + k' -m}-schedule on this instance and does not output FAIL.

Except for the recursion, all of Steps~1-4 can clearly be performed in $\bigO (n\cdot k)$.
Since we sort the clients by the processing time of their job, calling the recursion in Step~2 can be done in constant time, as only the large jobs (which are first up to three jobs) need to be removed from the instance and $k'$ and $d$ need to be adjusted.
Thus, a total running time of $\bigO (n (k + \log n))$ follows.
\end{proof}

Now we can combine Lemma~\ref{lem:apx-small-k} and Lemma~\ref{lem:apx-large-k} to get the following
result.

\begin{theorem}
Given a YES-instance $\mathcal{I} = (\{p_1, \ldots, p_n\}, m,d, k)$ of \ESSDabbrvstar, one can compute a solution to an instance $\mathcal{I} ' := (\{p_1, \ldots, p_n\}, m,d, k')$ with $\lfloor 2\frac{k}{3}\rfloor \ge k' \ge 2\lfloor \frac{k}{3}\rfloor$ in $O (n \cdot (k + \log n)) $ time.
\end{theorem}

We leave as an open question whether a similar result can be obtained for \ESSDabbrv.

\section{Precedence Constraints}
In this section, we investigate the computational complexity of \ESPC.

\subsection{Hardness Results}

The hardness result from Theorem~\ref{thm:sameduedatehardness} for \ESSDabbrvstar can
easily be adapted to \ESPCabbrvstar by modeling processing times by paths of
appropriate length in the precedence DAG. Hence, we get the following result.
\begin{corollary}
\ESPCabbrvstar is NP-hard and W[1]-hard when parameterized by the number of days $m$ even if $k=1$ and the precedence DAG consists of disjoint paths.
\end{corollary}
\begin{proof}[Proof sketch]
We use the same idea as in the proof of Theorem~\ref{thm:sameduedatehardness}. For each client~$i$ the reduction in the proof of Theorem~\ref{thm:sameduedatehardness} created one job of each day $j$ with processing time $p_i$ that can be encoded in unary. This allows us to introduce $p_i-1$ additional dummy clients for each client~$i$. In the precedence DAG for each day $j$ we add a directed path of length $p_i$ where the job of client~$i$ is the last one in the path and the other jobs in the path are the ones of the dummy clients for client~$i$ in some arbitrary order. This means that of the job of client~$i$ to be scheduled on a day $j$, the jobs of the dummy clients for client~$i$ also have to be scheduled on that day which simulated the processing time~$p_i$.
\end{proof}

For the setting where we do not have that all days look the same, we get
NP-hardness even for two days.
\begin{theorem}\label{thm:constdays}
\ESPCabbrv is NP-hard even if ${k=1}$ and ${m=2}$.
\end{theorem}
\begin{proof}
We reduce from \textsc{Clique}, where given a graph $H=(U,F)$ and an integer
$h$, we are asked to decide whether $H$ contains a complete subgraph with $h$
vertices. This problem is known to be NP-complete~\cite{Kar72}. Given a graph
$H=(U,F)$ and an integer $h$ we construct an instance of \ESPCabbrv as follows.
Assume that the vertices in $U$ are ordered in some fixed but arbitrary way,
that is, $U=\{v_1, v_2, \ldots, v_{|U|}\}$.
\begin{itemize}
\item For each vertex $v\in U$ we create one ``vertex client''~$i_v$ and for each
edge $e\in F$ we create one ``edge client''~$j_e$.
\item For day one we create the precedence DAG $G_1$ where for all
$\ell\in[|U|-1]$ we have that $(i_{v_\ell},i_{v_{\ell+1}})\in E_1$ and for all
$e\in F$ we have that $(i_{v_{|U|}},j_e)\in E_1$. That is, the precedence DAG is
a directed path containing all jobs of vertex clients and all
jobs of edge clients are out-neighbors of the job of the last vertex client in
the path.
Furthermore, we set the deadline~$d_1$ for day one to $|U|+|F|-\binom{h}{2}$.
\item For day two we create the precedence DAG $G_2$ where for all
$v\in U$ we have that $(i_v,j_e)\in E_2$ if and only if $v\in
e$. That is, for each edge $e$ of~$H$ the precedence DAG contains two arcs from
the jobs of the vertex clients corresponding to the endpoints of~$e$ to the
job of the edge client corresponding to~$e$.
Furthermore, we set the deadline $d_2$ for day two to~$h+\binom{h}{2}$.
\item We set $k=1$.
\end{itemize}
Clearly, the reduction can be computed in polynomial time. Intuitively, day one
is a ``selection gadget''.
The deadline and the precedence DAG are chosen in a way such that all jobs except the ones of $\binom{h}{2}$
edge clients can be scheduled. The second day is a ``validation gadget''
that ensures that the edges corresponding to the $\binom{h}{2}$
edge clients that have no job scheduled on day one form a clique in $H$.

$(\Rightarrow)$: Assume that $H$ contains a clique $X\subseteq U$ of size~$h$.
On day one, we schedule all jobs of vertex clients~$i_v$ in the order
prescribed by the precedence DAG $G_1$. Then for all edges~$e\in F$ such that not both endpoints
of~$e$ are in~$X$ we schedule the job of the corresponding edge client~$j_e$. Note that $G_1$ allows us to do this, since all jobs of vertex
clients are already scheduled. Furthermore, the deadline of day one allows us
to schedule jobs of all but $\binom{h}{2}$ clients. Since the vertices of $X$
form a clique, there are exactly $\binom{h}{2}$ edges that have both their
endpoints in $X$. Hence all jobs that are scheduled on day one finish before the
deadline. 

On day two, we first schedule all jobs of vertex clients $i_v$
with $v\in X$. Then we schedule the jobs of edge clients $j_e$ with
$e\subseteq X$, that is, both endpoints of~$e$ are part of the clique $X$. Note
that $G_2$ allows us to schedule the jobs of these edge clients since we
already scheduled the jobs of the vertex clients corresponding to the
endpoints the edges corresponding to the jobs of these edge clients. Note that
those edge clients are exactly the ones that do not have their jobs scheduled on day one. Furthermore,
the total number of jobs scheduled on day two is $h+\binom{h}{2}$, hence they
all finish before the deadline. It follows that we have found a set
of 1-equitable schedules.

$(\Leftarrow)$: Assume that there is a set of 1-equitable schedules. Note that
on day one, the precedence DAG required that the jobs of all vertex clients are
scheduled first, and then an arbitrary set of $|F|-\binom{h}{2}$ jobs of edge
clients can be scheduled. Let $F^\star\subseteq F$ be the set of edges such
that the corresponding edge clients do not have a job scheduled on day one. Note
that $|F^\star|\ge \binom{h}{2}$ and that all edge clients corresponding to
edges in $F^\star$ have their job scheduled on day two, otherwise the
set of schedules would not be 1-equitable. The precedence DAG $G_2$ for day two
requires that if a job of an edge client $j_e$ is scheduled, the jobs of the vertex
clients corresponding to the endpoints of $e$ need to be scheduled before. The
deadline of day two allows for at most $h$ additional jobs to be scheduled,
hence there need to be $h$ jobs that can be scheduled on day two such that
 all precedence constraints for the jobs of edge clients
 corresponding to edges in~$F^\star$ are fulfilled. Note that by construction of
 $G_2$ we can assume that all additionally scheduled jobs belong to vertex
 clients. Let $U^\star\subseteq U$ be the set of vertices corresponding to
 vertex clients that have a job scheduled on day two. We already argued that
 $|U^\star|\le h$. However, we also have that $|U^\star|\ge h$ since otherwise,
 by the pidgeon hole principle, there is at least one edge client
 corresponding to an edge in $F^\star$ that does not have the precedence
 constraints of its job fulfilled. It follows that $|U^\star|= h$ which implies
 that the vertices in $U^\star$ form a clique in $H$.
\end{proof}

In the following, we present some hardness results that show that even further restrictions on the precedence DAG presumably cannot yield polynomial-time solvability.

\begin{theorem}\label{thm:precedencefewpaths}
\ESPCabbrv is NP-hard even if $k=1$, $d=3$, and the precedence DAG of each day consists of at most two disjoint paths.
\end{theorem}

\begin{proof}
We reduce from the restriction of \textsc{Monotone Not-All-Equal-Sat}, where every variable appears in exactly three clauses, every clause contains two or three variables, and every clause contains only non-negated literals.
Given a set of clauses, we are asked whether there is an assignment of truth
values to the variables such that every clause contains at least one variable
that is set to true and at least one variable that is set to false. This
problem is known to be NP-complete~\cite{DBLP:journals/gc/DehghanSA15}.
By Hall's Marriage Theorem~\cite{Hall35}, the incidence graph contains a matching~$M$ which leaves no variable unmatched.
By considering the clause to which $M$ matches to a given variable~$x$ as the last occurence of~$x$, we may assume that the last occurrence of variable~$x$ is the last variable of the clause.
Let $a$ be the number of variables, $b_2$ be the number of clauses with two variables, and $b_3$ the number of clauses with three variables. We construct an instance of \ESPCabbrv as follows.
\begin{itemize}
\item We set the deadline to three, i.e., $d=3$, and we set $k=1$.
\item For each variable $x_j$, we create six clients: $i_1^{(j,T)}$, $i_2^{(j,T)}$, $i_3^{(j,T)}$, $i_1^{(j,F)}$, $i_2^{(j,F)}$, $i_3^{(j,F)}$.
\item We create three ``dummy clients'' $i_1^{(D)}$, $i_2^{(D)}$, $i_{3}^{(D)}$.
\item We create $m=1+a+b_2+2b_3$ days: one ``dummy day'', $a$ variable days,
and $b_2+2b_3$ clause days.
\end{itemize}
\emph{Dummy Day:} For the first day we create a precedence DAG that is one directed path starting with jobs of clients~$i_1^{(D)},i_2^{(D)},i_{3}^{(D)}$ and then the jobs of all remaining clients in an arbitrary order.

\smallskip

\noindent \emph{Variable Days:} For variable $x_j$ we create Day~$j+1$ with a precedence DAG that consists of two directed paths. The first path contains jobs of clients~$i_1^{(j,T)},i_2^{(j, T)},i_3^{(j,T)}$ in that order. The second path starts with jobs of clients $i_1^{(j,F)},i_2^{(j, F)},i_3^{(j,F)}$ in that order and then the jobs of all remaining clients in an arbitrary order.
\smallskip

\noindent \emph{Clause Days:}
Let~$(x_{j_1}, x_{j_2})$ be the $j$th clause containing two variables and let it contain the $t_1$th and $t_2$th appearence of $x_{j_1}$ and $x_{j_2}$, respectively.
We create Day~$a+j+1$ with a precedence DAG containing of the two paths $i_1^{(D)}, i_{t_1}^{(j_1, T)}, i_{t_2}^{(j_2, F)}$ and $i_2^{(D)},  i_{t_1}^{(j_1, F)}, i_{t_2}^{(j_2, T)}$.
Let $(x_{j_1},x_{j_2},x_{j_3})$ be the $j$th clause containing three variables and let it contain the $t_1$th, $t_2$th, and $t_3$th appearance of $x_{j_1}$, $x_{j_2}$, and $x_{j_3}$, respectively.
We create Days $a+b_2 + 2j$ and $a+ b_2 + 2j+ 1$ with precedence DAGs consisting of two directed paths.
On Day $a+ b_2 + 2 j$, the first path contains jobs of clients $i_{t_1}^{(j_1, F)}, i_{t_2}^{(j_2, F)}, i_{t_3}^{(j_3, T)}$ in that order.
The second path contains jobs of clients $i_{t_1}^{(j_1, T)}, i_{t_2}^{(j_2, T)},i_{t_3}^{(j_3, F)}$ in that order.
On Day $a + b_2 + 2j + 1$,
the first path starts with jobs of clients~$i_{1}^{(D)}, i_{t_1}^{(j_1, F)}, i_{t_2}^{(j_2, T)}$ in that order and then the jobs of all remaining clients in an arbitrary order.
The second path contains jobs of clients $i_2^{(D)}, i_{t_1}^{(j_1, T)}, i_{t_2}^{(j_2, F)}$ in that order.

This finishes the construction. 
Since maximum matchings on bipartite graphs can be computed in polynomial time~\cite{hopcroft1973n}, the reduction runs in polynomial time.

$(\Rightarrow)$
Assume we have a satisfying assignment of the \textsc{Monotone Not-All-Equal-SAT} formula.
We produce a set of $k$-equitable schedules for $k=1$ as follows.
On Day~1, we schedule the jobs of clients~$i_1^{(D)}, i_2^{(D)}, i_3^{(D)}$.
On each variable day, we schedule either jobs of clients $ i_1^{(j, T)}, i_2^{(j, T)}, i_3^{(j, T)}$ or jobs of clients~$i_1^{(j, F)}, i_2^{(j, F)}, i_3^{(j, F)}$.
We do the latter if variable~$x_j$ is set to true and the former otherwise.
For each clause containing exactly two variables~$x_{j_1}$ and~$x_{j_2}$, we schedule jobs of clients $i_1^{(D)},  i_{t_1}^{(j_1, T)}, i_{t_2}^{(j_2, F)}$ if $x_{j_1}$ is set to true, while we schedule jobs of clients~$i_2^{(D)}, i_{t_1}^{(j_1, F)}, i_{t_2}^{(j_2, T)}$ otherwise.
For each clause containing three variables~$x_{j_1}$, $x_{j_2}$, and $x_{j_3}$, we schedule jobs of clients~$i_{t_1}^{(j_1, F)}, i_{t_2}^{(j_2, F)}, i_{t_3}^{(j_3, T)}$ on the first clause day if $x_{j_3}$ is set to true and jobs of clients~$i_{t_1}^{(j_1, T)}, i_{t_2}^{(j_2, T)}, i_{t_3}^{(j_3, F)}$ otherwise.
On the second clause day, we schedule the jobs of clients~$i_{1}^{(D)}, i_{t_1}^{(j_1, F)}, i_{t_2}^{(j_2, T)}$ if either~$x_{j_1}$ is set to false and $x_{j_3}$ is set to false or $x_{j_2}$ is set to true and $x_{j_3}$ is set to true.
Otherwise, we schedule the jobs of clients~$i_2^{(D)}, i_{t_1}^{(j_1, T)}, i_{t_2}^{(j_2, F)}$.

It is easy to verify that all clients have at least on of their jobs scheduled.

$(\Leftarrow)$
Assume that we have a $k$-equitable set of schedules for $k=1$.
On Day~1, we may assume that jobs of clients~$i_1^{(d)},i_2^{(d)},i_3^{(d)}$ are scheduled, since it is never beneficial to leave slots empty.
First we show that one each variable day either jobs of clients~$ i_1^{(j, T)}, i_2^{(j, T)}, i_3^{(j, T)}$ or jobs of clients~$ i_1^{(j, F)}, i_2^{(j, F)}, i_3^{(j, F)}$ are scheduled:
Assume for a contradiction that this is not true for variable $x_j$.
Since the third appearance of~$x_j$ is the last variable of the corresponding clause, either job of client~$i_{t_3}^{(j_3, T)}$ or $i_3^{(j_3, F)}$ is not scheduled, a contradiction.
We claim that setting variable~$x_j$ to true if and only if jobs of clients~$i_1^{(j, F)}, i_2^{(j, F)}, i_3^{(j, F)}$ are scheduled on Day~$j+1$ yields a satisfying assignment.

Consider the $j$th clause~$(x_{j_1}, x_{j_2})$ with two variables, containing the $t_1$th and $t_2$th appearance of $x_{j_1}$ and $x_{j_2}$, respectively.
By the precedence constraints, jobs of clients $i_{t_1}^{(j_1, T)}, i_{t_1}^{(j_1, F)}, i_{t_2}^{(j_2, T)},i_{t_2}^{(j_2, F)}$ can only be scheduled on Day~${j_1 + 1}$, Day~$j_2 + 1$ or Day~$a+1 + j$.
If~$x_{j_1}$ and~$x_{j_2}$ are both set to true, then the jobs of both $i_{t_1}^{(j_1, T)}$ and~$i_{t_2}^{(j_2, T)}$ need to be scheduled on Day~$a + 1+ j$, which is impossible.
Similarly, if $x_{j_1}$ and $x_{j_2}$ are both set to false, then the jobs of both $i_{t_1}^{(j_1, F)}$ and $i_{t_2}^{(j_2, F)}$ need to be scheduled on Day~$a+ 1+ j$, which is impossible.

Consider the $j$th clause~$(x_{j_1}, x_{j_2}, x_{j_3})$ with three variables, containing the $t_1$th, $t_2$th, and $t_3$th appearance of $x_{j_1}$, $x_{j_2}$, and $x_{j_3}$, respectively.
By the precedence constraints, jobs of clients $i_{t_1}^{(j_1, T)}, i_{t_1}^{(j_1, F)}, i_{t_2}^{(j_2, T)},i_{t_2}^{(j_2, F)},i_{t_3}^{(j_3, T)}, i_{t_3}^{(j_3, F)}$ can only be scheduled on Day~$j_1 + 1$, Day~$j_2 + 1$, Day~$j_3 + 1$, Day $a+ b_2 + 2j$ or Day~$a + b_2 + 2j + 1$.
If $x_{j_1}, x_{j_2}$, and $x_{j_3}$ are set to true, then the job of client~$i_{t_3}^{(j_3, T)}$ needs to be scheduled on the first clause day.
Thus, jobs of both $i_{t_1}^{(j_1, T)}$ and $i_{t_2}^{(j_2, T)}$ need to be scheduled on the second clause day, which is impossible.
The case that~$x_{j_1}, x_{j_2}$, and $x_{j_3}$ are set to false leads to a contradiction by symmetric arguments.
\end{proof}

We remark that by introducing additional dummy clients, the reduction for
Theorem~\ref{thm:precedencefewpaths} can be modified in a way that the precedence DAGs consists of disjoint paths of constant length.
Hence, we get the following result.

\begin{corollary}
\ESPCabbrv is NP-hard even if $k=1$ and the precedence DAG of each day consists of disjoint paths of length at most four.
\end{corollary}
\begin{proof}[Proof sketch]
In the proof of Theorem~\ref{thm:precedencefewpaths} we have that the precedence DAG consists of two disjoint paths on each day. However the deadline for each day is $d=3$, which means that every job that in located in one of the paths at a distance larger than three from the source of the path cannot be scheduled on that day. We can achieve the same by introducing at most $3 \cdot n$ dummy clients and $ n$ additional days (where $n$ is the number of clients in the instance constructed in the proof of Theorem~\ref{thm:precedencefewpaths}). On the new days we make sure that all jobs of the dummy clients can be scheduled. That can be done by using a precedence DAG that consist of disjoint paths of length four where the first three jobs of each path stem from dummy clients and the last job of each path (which cannot be scheduled before the deadline) stems from an original client. For each of the original days we modify the precedence DAG in the following way. We replace each of the disjoint paths by its first four vertices.  For each job that was located at a distance larger than four from the source of the path, we introduce a new path of length four to the precedence DAG that starts with three jobs of dummy clients and then the job from the original client. This ensures that this job now also cannot be scheduled on that day. Note that we have introduced sufficiently many dummy clients to be able to do this. The jobs of the (potentially) remaining dummy clients are distributed arbitrarily into disjoint paths of length at most four. Note that now all precedence DAGs consist of disjoint paths of length at most four.
\end{proof}

\begin{table*}[h]
 \begin{align}
   \sum_{A' \subseteq A} x_{G, A', d}  = \gamma (G, d) & \qquad \forall d\in \{n-|A| + 1, \dots, n\}, \text{ precedence DAGs } G \label{eq:prec-graphs-d}\\
   \sum_{A' \subseteq A} \xspecial_{G, A'} = \gammaspecial (G) & \qquad \forall \text{ precedence DAGs } G \label{eq:prec-graphs}\\
   \sum_{A' : i\in A'} \sum_{G} \Bigl( \xspecial_{G, A'} + \sum_{d =n - \alpha + 1}^{n} x_{G, A', d} \Bigr) \ge k & \qquad \forall i\in A \label{eq:jobs-from-A}\\
   x_{G, A', d} = 0 & \qquad \text{if } |A'| > d \label{eq:no-overloaded-day}\\
   \sum_{A' \subseteq A: |A'| \le d} \xspecial_{G, A'}  \ge \gamma^{\le d} (G) & \qquad \forall d\in \{1, \dots, n- \alpha\} \label{eq:no-overloaded-day-2}\\
   x_{G, A'}  = 0 & \qquad \forall A', G : \exists (i, i') \in E(G) \text{ with } i\notin A' \land i'\in A' \label{eq:obey-precedence}\\
   x_{G, A', d}  = 0 & \qquad \forall A', G, d\in \{n-\alpha +1, \dots, n\} : \exists (i, i') \in E(G) \nonumber \\
   & \qquad \text{with } i\notin A' \land i'\in A' \label{eq:obey-precedence2}\\
   \noalign{$\sum_{A' \subseteq A, G, d\in \{n - \alpha  +1 , \dots, n\}} \min\{d - |A'|, n -
   \alpha \} \cdot x_{G, A', d} + \sum_{j \in \{1, \dots, m\}: d_j \le n- \alpha} d_j - \sum_{A' \subseteq A, G} |A'| \cdot \xspecial (G, A')
   \ge k (n - |A|)$} \label{eq:other-jobs}
 \end{align}
 \end{table*}

\subsection{Algorithmic Result} In the following, we give an ILP formulation
for \ESPCabbrv to obtain fixed-parameter tractability for the number of clients that are incident to an arc in at least on precedence DAG.

\begin{theorem}
\ESPCabbrv is fixed-parameter tractable when parameterized by the number of clients that are incident to an arc in at least on precedence DAG.
\end{theorem}

\begin{proof}
 Let $\mathcal{I}$ be an instance of \ESPCabbrv.
 We assume without loss of generality that $d_j \le n$ for all days $j \in \{1, \dots, m\}$ (since on every day at most $n$ jobs can be scheduled, we can replace the deadline by $n$ otherwise).
 Let $A$ be the set of clients incident to at least one arc appearing in some precedence DAG.
 Let $\alpha := |A| $ and $\beta$ be the number of arcs appearing in at least one incidence DAG.
 Note that $\frac{\alpha}{2} \le \beta \le \binom{\alpha}{2}$.
 Note that the number of different precedence DAGs is at most $2^\beta$.
 For a precedence DAG $G$ and a deadline $d\in \{1, \dots, m\}$, let $\gamma (G, d)$ denote the number of days with precedence DAG $G$ and deadline $d$.
 We define $\gamma (G): = \sum_{r = 1}^n \gamma (G, r)$ and $\gamma^{\le d} (G) := \sum_{r =1 }^{d} \gamma (G, r)$.

 We construct an integer linear program (ILP) as follows.
 For each precedence DAG~$G$, subset $A' \subseteq A$, and $d\in \{n- \alpha + 1, \dots, n\}$, we add a variable $x_{G, A', d}$, indicating on how many days with precedence DAG~$G$ and deadline~$d$ exactly the jobs from clients in $A'$ are scheduled.
 Additional, for each precedence DAG~$G$ and subset $A' \subseteq A$, we add a variable $\xspecial_{G, A'}$, indicating on how many days with precedence DAG~$G$ and deadline at most $n-\alpha$ the jobs from clients in~$A'$ are scheduled.
 Furthermore, there are the constraints specified by Equations~\ref{eq:prec-graphs-d},~\ref{eq:prec-graphs},~\ref{eq:jobs-from-A},~\ref{eq:no-overloaded-day},~\ref{eq:no-overloaded-day-2},~\ref{eq:obey-precedence},~\ref{eq:obey-precedence2}, and~\ref{eq:other-jobs}.

 The number of variables in this ILP is bounded by $\alpha 2^{\alpha + \beta}$ and therefore can be solved in FPT-time with respect to $\alpha + \beta$ by~\citet{lenstra1983integer}.
 We now show that any solution to this ILP corresponds to a feasible schedule, and each solution to $\mathcal{I}$ corresponds to a solution to the ILP.

 $(\Rightarrow)$:
 Let $x$ be a feasible solution to the ILP.
 For each precedence DAG $G$ and each $d\in \{n - \alpha +1, \dots, n\}$, we schedule on $x_{G, A', d}$ of the days with precedence DAG~$G$ and deadline~$d$ the jobs of clients contained in $A'$, and no job of a client from $A\setminus A'$.
 For each precedence DAG~$G$, we schedule on $\xspecial_{G, A'}$ of the $\gammaspecial (G)$ days with precedence DAG $G$ and deadline at most $n - \alpha$ the jobs of clients contained in $A'$, and no job of a client from $A \setminus A'$.
 We do this in such a way that for each $A', A'' \subseteq A$ with $|A'| < |A''|$, it holds that the deadline of the days on which $A''$ is scheduled is at least the deadline of the days on which $A'$ is scheduled.
 We iterate from Day 1 to $m$.
 Let $A_j \subseteq A$ be the set of clients which have already a job scheduled on Day~$j$.
 For each Day $j$, as long as there are less than $d_j$ jobs scheduled on this day, and there is one job from a client in $\{1,\ldots,n\}\setminus A$ which has not been scheduled on this day, we pick a client~$i \in \{1, \dots, n\} \setminus A$ whose job has been scheduled fewest time till now, and schedule its job on Day $j$.
 Note that this procedure ensures that for two jobs~$i, i'\in \{1, \dots, n\} \setminus A$, during any point of the procedure, job~$i$ is scheduled at most once more than job $i'$.
 Thus, no job is scheduled twice on one day.

 Equations~\ref{eq:prec-graphs-d} and~\ref{eq:prec-graphs} ensure that each precedence DAG is considered exactly the number of times it actually appears in $\mathcal{I}$.
 Equation~\ref{eq:jobs-from-A} ensures that every client in $A$ has at least $k$ of its jobs scheduled, while Equation~\ref{eq:other-jobs} ensures this for all clients in $\{1,\ldots,n\}\setminus A$:
 On a day~$j$ with $d_j \in \{n-\alpha +1, \dots, n\}$, we schedule $\min\{d_j - |A'|, n - |A|\}$ jobs of these clients, while on a day~$j$ with $d_j \le n - \alpha$, we schedule $d_j - |A'_j|$ jobs of these clients.
 Since jobs of client~$i$ are scheduled at most once more than jobs of clients $i'$, it follows that the jobs of every client are scheduled at least $k$ times.
 Equations~\ref{eq:obey-precedence} and~\ref{eq:obey-precedence2} ensure that the precedence DAG is obeyed on every day.
 Equation~\ref{eq:no-overloaded-day} ensures that for every day with deadline at least $n - \alpha + 1$, the jobs scheduled on this day can be performed before the deadline.
 For all days with deadline at most~$n - \alpha$, this is ensured by Equation~\ref{eq:no-overloaded-day-2}:
 This inequality ensures that for any precedence graph~$G$ and any $d\le n-\alpha$, there are at least $\gamma^{\le d} (G) $ many days on which at most $d$ jobs from~$A $ are scheduled, and consequently, we schedule at most $d_j$ jobs on each such day $j$.

 $(\Leftarrow)$:
 Consider a feasible schedule.
 We set $x_{G, A', d}$ to be the number of days with precedence DAG $G$ and deadline $d \in \{n- \alpha +1 , \dots, n\}$ on which all jobs from clients in~$A'$ but no job from a client in $A\setminus A'$ is scheduled.
 Thus, Equation~\ref{eq:prec-graphs-d} is fulfilled.
 Similarly, let $\xspecial_{G, A'}$ be the number of days with precedence DAG $G$ and deadline at most $n - \alpha$ on which all jobs from clients in~$A'$ but no job from a client in $A\setminus A'$ is scheduled.
 It follows that Equation~\ref{eq:prec-graphs} is fulfilled.
 Since every client in $A$ has at least $k$ of its jobs scheduled, also Equation~\ref{eq:jobs-from-A} is fulfilled.
 Since for any day $j$ with deadline $d_j \in \{n - \alpha +1, \dots, n\}$ there are at most $d_j$ jobs (from~$A$) scheduled on this day, it follows that Equation~\ref{eq:no-overloaded-day} is fulfilled.
 Similarly, for each~$d\in \{1, \dots, n - \alpha\}$, on each of the days with deadline $d$ and precedence graph $G$, the set $A'$ of clients from $A$ whose jobs are scheduled on this day fulfills $|A'| \le d$, and therefore, Equation~\ref{eq:no-overloaded-day-2} holds.
 By the precedence constraints also Equations~\ref{eq:obey-precedence} and~\ref{eq:obey-precedence2} are fulfilled.
 For any day $j \in \{1, \dots, m\}$, let $|A'_j|$ be the set of clients from $A$ whose jobs have been scheduled on Day $j$.
 There are at most $\min \{d_j - |A'_j|, n - |A|\}$ jobs from $\{1, \dots, n\} \setminus A'$ scheduled on Day $j$.
 Since $\min \{d_j - |A'_j|, n- \alpha\} = d_j - |A'_j|$ if $d_j \le n - \alpha$, and every client in $\{1,\ldots,n\}\setminus A$ has at least $k$ of its jobs scheduled, it follows that Equation~\ref{eq:other-jobs} is fulfilled.
\end{proof}

\section{Conclusion}

We have introduced a promising new framework for single machine scheduling
problems. We investigated three basic single machine scheduling problems in this
framework and we believe that it might also be interesting in other scheduling
contexts. 

We leave several questions open for future research. We believe that it would be
promising to implement our approximation algorithm for \ESSDabbrvstar and, once provided with appropriate real-world data, test
how well it performs in practice. 
The question whether we can get similar
approximation results also for \ESSDabbrv and \ESPCabbrv remains unresolved.
For \ESPCabbrv, it is also remains open whether we can get similar combinatorial
algorithms as for \ESSDabbrv.

\bibliographystyle{abbrvnat}
\bibliography{journal/bibliography}

\end{document}